\documentclass[a4paper,UKenglish]{lipics-v2019}

\usepackage{microtype}%if unwanted, comment out or use option "draft"
\usepackage{algorithmicx}
\usepackage{algpseudocode}
\usepackage{url}
\usepackage{booktabs}

\graphicspath{{./graphics/}} %helpful if your graphic files are in another directory

\title{Compressed Multiple Pattern Matching}

\author{Dmitry Kosolobov}{University of Helsinki, Helsinki, Finland}{dkosolobov@mail.ru}{}{}%mandatory, please use full name; only 1 author per \author macro;
\author{Nikita Sivukhin}{Ural Federal University, Ekaterinburg, Russia}{sivukhin.nikita@yandex.ru}{}{}

\authorrunning{D. Kosolobov and N. Sivukhin} %mandatory. First: Use abbreviated first/middle names. Second (only in severe cases): Use first author plus 'et. al.'

\Copyright{Dmitry Kosolobov and Nikita Sivukhin}%mandatory, please use full first names. LIPIcs license is "CC-BY";  http://creativecommons.org/licenses/by/3.0/


% mandatory
\relatedversion{\url{https://arxiv.org/abs/1811.01248}}

\nolinenumbers
\newcommand*{\istr}[1]{\ensuremath{\mathsf{istr}(#1)}}
\newcommand*{\trans}[2]{\ensuremath{\textsf{next}(#1, #2)}}
\newcommand*{\transdef}{\ensuremath{\textsf{next}}}

\newcommand*{\dict}{\ensuremath{\mathcal{D}}}
\newcommand*{\tree}[1]{\ensuremath{\mathcal{T}}}
\newcommand*{\treedef}{\ensuremath{\mathcal{T}}}
\newcommand*{\str}[1]{\ensuremath{\textsf{str}(#1)}}
\newcommand*{\flag}[1]{\ensuremath{\textsf{mark}(#1)}}
\newcommand*{\flagdef}{\ensuremath{\textsf{mark}}}
\newcommand*{\link}[1]{\ensuremath{\textsf{failure}(#1)}}
\newcommand*{\linkdef}{\ensuremath{\textsf{failure}}}
\newcommand*{\report}[1]{\ensuremath{\textsf{report}(#1)}}
\newcommand*{\reportdef}{\ensuremath{\textsf{report}}}

\newcommand*{\num}[1]{\ensuremath{\textsf{num}(#1)}}
\newcommand*{\numdef}{\ensuremath{\textsf{num}}}
\newcommand*{\numdictdef}[0]{\ensuremath{\textsf{B}}}

\newcommand*{\numdictgen}[1]{\ensuremath{\textsf{B}_{#1}}}

\makeatletter
\newcommand\footnoteref[1]{\protected@xdef\@thefnmark{\ref{#1}}\@footnotemark}
\makeatother

\funding{Supported by the Russian Science Foundation (RSF), project 18-71-00002.}

\EventEditors{Nadia Pisanti and Solon P. Pissis}
\EventNoEds{2}
\EventLongTitle{30th Annual Symposium on Combinatorial Pattern Matching (CPM 2019)}
\EventShortTitle{CPM 2019}
\EventAcronym{CPM}
\EventYear{2019}
\EventDate{June 18--20, 2019}
\EventLocation{Pisa, Italy}
\EventLogo{}
\SeriesVolume{128}
\ArticleNo{10}
\begin{document}

\maketitle

\begin{abstract}
Given $d$ strings over the alphabet $\{0,1,\ldots,\sigma{-}1\}$, the classical Aho--Corasick data structure allows us to find all $occ$ occurrences of the strings in any text $T$ in $O(|T| + occ)$ time using $O(m\log m)$ bits of space, where $m$ is the number of edges in the trie containing the strings. Fix any constant $\varepsilon \in (0, 2)$. We describe a compressed solution for the problem that, provided $\sigma \le m^\delta$ for a constant $\delta < 1$, works in $O(|T| \frac{1}{\varepsilon} \log\frac{1}{\varepsilon} + occ)$ time, which is $O(|T| + occ)$ since $\varepsilon$ is constant, and occupies $mH_k + 1.443 m + \varepsilon m + O(d\log\frac{m}{d})$ bits of space, for all $0 \le k \le \max\{0,\alpha\log_\sigma m - 2\}$ simultaneously, where $\alpha \in (0,1)$ is an arbitrary constant and $H_k$ is the $k$th-order empirical entropy of the trie. Hence, we reduce the $3.443m$ term in the space bounds of previously best succinct solutions to $(1.443 + \varepsilon)m$, thus solving an open problem posed by Belazzougui. Further, we notice that $L = \log\binom{\sigma (m+1)}{m} - O(\log(\sigma m))$ is a worst-case space lower bound for any solution of the problem and, for $d = o(m)$ and constant~$\varepsilon$, our approach allows to achieve $L + \varepsilon m$ bits of space, which gives an evidence that, for $d = o(m)$, the space of our data structure is theoretically optimal up to the $\varepsilon m$ additive term and it is hardly possible to eliminate the term $1.443m$. In addition, we refine the space analysis of previous works by proposing a more appropriate definition for $H_k$. We also simplify the construction for practice adapting the fixed block compression boosting technique, then implement our data structure, and conduct a number of experiments showing that it is comparable to the state of the art in terms of time and is superior in space.
\end{abstract}

\section{Introduction}

Searching for multiple patterns in text is a fundamental stringology problem that has numerous applications, including bioinformatics~\cite{Gusfield}, search engines~\cite{WuManber}, intrusion detection systems~\cite{LiaoLinLinTung,PaoEtAl}, shortest superstring approximation~\cite{AlankoNorri}, and others. The classical solution for the multiple pattern matching is the Aho--Corasick data structure~\cite{AhoCorasick}, which, however, does not always fulfil space requirements of many such applications due to the rapid growth of the amounts of data in modern systems. To address this issue, several space-efficient multiple pattern matching data structures were developed in the last decade. In this paper we improve the space consumption in a state-of-the-art solution for the problem, simplify the compression method used in this solution by adapting the known fixed block compression boosting technique, and give an evidence that the achieved space is, in a sense, close to optimal; in addition, we refine the theoretical space analysis of a previous best result, and implement our construction and conduct a number of experiments showing that it is comparable to the existing practical data structures in terms of time and is superior in space. Before discussing our contribution in details, let us briefly survey known results in this topic.

The Aho--Corasick solution builds a trie of all patterns augmented with additional structures overall occupying $O(m\log m)$ bits (hereafter, $\log$ denote logarithm with base 2), where $m$ is the number of edges in the trie, and allows us to find all $occ$ occurrences of the patterns in any text $T$ in $O(|T| + occ)$ time. The $O(m\log m)$-bit space overhead imposed by this solution might be unacceptably high if one is processing large sets of patterns; for this case, several succinct and compressed data structures were developed in the last decade~\cite{Belazzougui,ChanHonLamSadakane,FeigenblatPoratShiftan,HonEtAl,HonEtAlOld,KopelowitzPoratRozen,TamWuLamYiu}. In this paper we are especially interested in the two closely related results from~\cite{Belazzougui} and~\cite{HonEtAl}, which provide currently the best time and space bounds for the multiple pattern matching problem. In~\cite{Belazzougui} Belazzougui designed a compact representation of the Aho--Corasick scheme that works in the same $O(|T| + occ)$ time but stores the trie with all additional structures in only $m\log\sigma + 3.443 m + o(m) + O(d\log\frac{m}{d})$ bits, where $\sigma$ is the alphabet size and $d$ is the number of patterns; in addition, he showed that the space can be improved to $mH_0 + 3.443 m + o(m) + O(d\log\frac{m}{d})$ bits with no slowdown provided $\sigma \le m^\delta$ for a constant $\delta < 1$, where $H_0$ is the zeroth-order empirical entropy of the trie. In~\cite{HonEtAl} Hon et al.~further lowered the space to $m H_k + 3.443 m + o(m) + O(d\log\frac{m}{d})$ bits (again, assuming $\sigma \le m^\delta$) by simply applying the compression boosting technique~\cite{FerraginaEtAlBoosting}, where $H_k$ is the $k$th-order empirical entropy of the trie (see clarifications below) and $k$ is any fixed integer such that $0 \le k \le \alpha\log_\sigma m - 1$, for arbitrary constant $\alpha \in (0,1)$. This topic is rich with other related results, which, for instance, support dynamic modifications of the patterns \cite{ChanHonLamSadakane,FeigenblatPoratShiftan2,FeigenblatPoratShiftan}, try to process $T$ in real-time~\cite{KopelowitzPoratRozen}, allow randomization~\cite{CliffordEtAl,GolanPorat}, consider hardware implementations~\cite{DimopoulosPapaefstathiouPnevmatikatos}, etc. \mbox{In this paper we focus on the basic functionality as in \cite{Belazzougui} and \cite{HonEtAl}}.

Belazzougui posed the following open problem~\cite{Belazzougui}: can we reduce the constant $3.443$ in the space of his (and Hon et al.'s) result without any significant slowdown? We solve this problem affirmatively describing a data structure that, provided $\sigma \le m^\delta$ for a constant $\delta < 1$, occupies $m H_k + 1.443 m + \varepsilon m + O(d\log\frac{m}{d})$ bits of space, for an arbitrarily chosen constant $\varepsilon \in (0, 2)$, and answers pattern matching queries on any text $T$ in $O(|T| \varepsilon^{-1} \log\varepsilon^{-1} + occ)$ time, which is $O(|T| + occ)$ since $\varepsilon$ is constant. Then, we notice that $L = \log\binom{\sigma(m+1)}{m} - O(\log(\sigma m))$ is a worst-case space lower bound for any multiple pattern matching data structure and, for $d = o(m)$ and constant $\varepsilon$, our construction allows to achieve $L + \varepsilon m$ bits of space; observe that, for $\sigma = \omega(1)$, we have $L = m\log\sigma + m\log e + o(m) \approx m\log\sigma + 1.443m + o(m)$ (see~\cite{Belazzougui}), which gives an evidence that, for $d = o(m)$, our space bound is optimal up to the $\varepsilon m$ additive term and it is hardly possible to remove the term $1.443m$. In addition, we argue that the definition of $H_k$ borrowed by Hon et al.~\cite{HonEtAl} from~\cite{FerraginaEtAlXBW}, denoted $H_k^*$ in our paper, is not satisfactory: in particular, $H^*_k$ can be greater than $\log\sigma$, which contradicts the idea of the empirical entropy ($H_k^*$ was devised in~\cite{FerraginaEtAlXBW} for a slightly different problem); we propose a more appropriate definition for $H_k$, which is not worse since $H_k \le H_k^*$, and refine the analysis of Hon et al.~showing that their data structure indeed occupies $m H_k + 3.443 m + o(m) + O(d\log\frac{m}{d})$ bits even according to our definition of $H_k$. Further, our solution, unlike~\cite{HonEtAl}, does not require to fix $k$ and the space bound holds for all $k$ such that $0 \le k \le \max\{0,\alpha\log_\sigma m{-}2\}$ simultaneously; this is achieved by adapting the fixed block compression boosting technique~\cite{GogEtAl,KarkkainenPuglisi} to our construction, which is also better for practice than the compression boosting used in~\cite{HonEtAl}. Finally, we implement our data structure and conduct a number of experiments showing that it is comparable to the state of the art in terms of time and is superior in space.

The paper is organized as follows. In the following section we introduce some basic notions and define the $k$th-order empirical entropy of tries. In Section~\ref{sec:basic_algorithm} we survey the solution of Belazzougui~\cite{Belazzougui}. In Section~\ref{sec:compression_boosting} we discuss compression boosting techniques and investigate the flaws of the space analysis of Hon et al.~\cite{HonEtAl}. Section~\ref{sec:space_reduction} describes our main data structure and considers space optimality. Appendix contains implementation details and experiments.

\section{Preliminaries}
Throughout the paper, we mainly consider strings drawn from the alphabet $\{0, 1, \ldots, \sigma{-}1\}$ of size $\sigma$ (not necessarily constant). For a string $s = c_0 c_1 \cdots c_{n - 1}$, denote by $|s|$ its length $n$. The \emph{reverse} $c_{n-1}\cdots c_1 c_0$ of $s$ is denoted $s^r$. We write $s[i]$ for the letter $c_i$ of $s$ and $s[i..j]$ for the \emph{substring} $c_i c_{i + 1} \cdots c_j$, assuming $s[i..j]$ is the empty string if $i > j$. We say that a string $p$ \emph{occurs} in $s$ at position $i$ if $s[i..i{+}|p|{-}1] = p$. A string $p$ is called a \emph{prefix} (resp., \emph{suffix}) of $s$ if $p$ occurs in $s$ at position $0$ (resp., $|s| - |p|$). For integer segments, we use the following notation: $[i..j] = \{i, i{+}1, \ldots, j\}$, $(i..j] = [i..j] \setminus \{i\}$, $[i..j) = [i..j] \setminus \{j\}$. The set of all strings of lengths $k$ over an alphabet $A$ is denoted $A^k$; we use this notation for the set $[0..\sigma]^k$. For letter $c$, the only string of the set $\{c\}^k$ is denoted $c^k$.

The trie containing a set of strings $S$ is the minimal in the number of vertices rooted tree with edges labeled by letters so that each $s \in S$ can be spelled out on the path from the root to a vertex. For vertex $v$, denote the string written on the path from the root to $v$ by $\str{v}$.

The \emph{zeroth-order empirical entropy} (see~\cite{KosarajuManzini,Manzini}) of a string $t$ of length $n$ is defined as $H_0(t) = \sum_{c \in [0..\sigma)} \frac{n_c}{n} \log \frac{n}{n_c}$, where $n_c$ is the number of letters $c$ in $t$ and $\frac{n_c}{n}\log\frac{n}{n_c} = 0$ whenever $n_c = 0$. For a string $w$ of length $k$, let $t_w$ be a string formed by concatenating all letters immediately following occurrences of $w$ in the string $\$^k t$, where $\$ = \sigma$ is a new special letter introduced for technical convenience; e.g., $t_{ab} = aac$ and $t_{\$a} = b$ for $t = abababc$. The \emph{$k$th-order empirical entropy} of $t$ is defined as $H_k(t) = \sum_{w \in [0..\sigma]^k} \frac{|t_w|}{n} H_0(t_w)$ (see~\cite{KosarajuManzini,MakinenNavarro,Manzini}). It is well known that $\log\sigma \ge H_0 \ge H_1 \ge \cdots$ and $H_k$ makes sense as a measure of string compression only for $k < \log_\sigma n$ (see~\cite{Gagie} for a deep discussion). For the sake of completeness, let us show that $H_k \ge H_{k+1}$; this proof then can be easily adapted for the empirical entropy of tries below. Curiously, to our knowledge, all sources refer to this simple and intuitive fact as ``obvious'' but do not give a proof; even the original paper~\cite{Manzini}, the survey~\cite{MakinenNavarro}, and the book \cite{Navarro}.

\begin{lemma}[{see \cite[Lemma 3]{KarkkainenPuglisi}}]
For any strings $t_1, t_2, \ldots, t_\ell$ and the string $t = t_1 t_2 \cdots t_\ell$, we have $|t|H_0(t) \ge \sum_{i=1}^\ell |t_i|H_0(t_i)$.\label{lemma:entropies}
\end{lemma}
%\begin{proof}
%It suffices to consider the case $t = t_1t_2$. Denote the number of letters $c$ in the strings $t$, $t_1$, $t_2$ by $n_c$, $n_{1,c}$, $n_{2,c}$, respectively. Let $n = |t|$, $n_1 = |t_1|$, $n_2 = |t_2|$. The Gibbs inequality states that $\sum_{c=0}^\sigma n_{i,c} \log\frac{n_i}{n_{i,c}} \le \sum_{c=0}^\sigma n_{i,c} \log\frac{n}{n_c}$, for $i = 1,2$. Hence, $|t_1| H_0(t_1) + |t_2| H_0(t_2) = \sum_{c=0}^\sigma n_{1,c} \log\frac{n_1}{n_{1,c}} + \sum_{c=0}^\sigma n_{2,c} \log\frac{n_2}{n_{2,c}} \le \sum_{c=0}^\sigma (n_{1,c} + n_{2,c})\log\frac{n}{n_c} = |t| H_0(t)$.
%\end{proof}

Since, without loss of generality, one can assume that $t_w = t_{a_0w} t_{a_1w} \cdots t_{a_\sigma w}$, where $a_0, \ldots, a_\sigma$ are all letters of $[0..\sigma]$, Lemma~\ref{lemma:entropies} directly implies the inequality $|t|H_k(t) = \sum_{w \in [0..\sigma]^k} |t_w| H_0(t_w) \ge \sum_{w \in [0..\sigma]^{k+1}} |t_w| H_0(t_w) = |t|H_{k+1}(t)$ and, hence, $H_k \ge H_{k+1}$.

By analogy, one can define the empirical entropy for tries (see also~\cite{FerraginaEtAlXBW}). Let $\mathcal{T}$ be a trie with $n$ edges over the alphabet $[0..\sigma)$. For a string $w$ of length $k$, denote by $\mathcal{T}_w$ a string formed by concatenating in an arbitrary order the letters on the edges $(u,v)$ of $\mathcal{T}$ (here $u$ is the parent of $v$) such that $w$ is a suffix of $\$^k\str{u}$; e.g., $\mathcal{T}_{\$^k}$ consists of all letters labeling the edges incident to the root. Then, the \emph{$k$th-order empirical entropy} of $\mathcal{T}$ is defined as $H_k(\mathcal{T}) = \sum_{w \in [0..\sigma]^k} \frac{|\mathcal{T}_w|}{n} H_0(\mathcal{T}_w)$. (Note that $\sum_{w \in [0..\sigma]^k} |\mathcal{T}_w| = n$.) Analogously to the case of strings, one can show that $\log\sigma \ge H_0(\mathcal{T}) \ge H_1(\mathcal{T}) \ge \cdots$ and $H_k(\mathcal{T})$ makes sense as a measure of compression only for $k < \log_\sigma n$. For the definition of the $k$th-order empirical entropy of tries as given by Hon et al.~\cite{HonEtAl}, see Section~\ref{sec:compression_boosting}.

\section{Basic Algorithm}
\label{sec:basic_algorithm}

Given a dictionary $\dict$ of $d$ patterns, the multiple pattern matching problem is to preprocess $\dict$ in order to efficiently find all occurrences of the patterns in an arbitrary given text. In this section we describe for this problem the classical Aho--Corasick solution~\cite{AhoCorasick} and its space-efficient version developed by Belazzougui~\cite{Belazzougui}.

The main component of the Aho--Corasick data structure is the trie $\tree{\dict}$ containing $\dict$. Each vertex $v$ of the trie is augmented with the following structures (see Figure~\ref{fig:automaton}):
\begin{itemize}
\item a flag $\flag{v}$ that indicates whether $\str{v} \in \dict$: it is so iff $\flag{v} = 1$;
\item a hash table $\trans{v}{\cdot}$ that, for each letter $c$, either maps $c$ to a vertex $u = \trans{v}{c}$ such that $\str{u} = \str{v}c$, or returns $\textsf{nil}$ if there is no such $u$;
\item a link $\link{v}$ to a vertex such that $\str{\link{v}}$ is the longest proper suffix of $\str{v}$ that can be spelled out on a root-vertex path ($\link{v}$ is undefined if $v$ is the root);
\item a link $\report{v}$ to a vertex such that $\str{\report{v}}$ is the longest proper suffix of $\str{v}$ that belongs to $\dict$, or $\report{v} = root$ if there is no such suffix.
\end{itemize}

It is well known that the described data structure allows us to find all $occ$ occurrences of the patterns from $\dict$ in text $T$ in $O(|T| + occ)$ time: we read $T$ from left to right maintaining a ``current'' vertex $v$ in $\tree{\dict}$ (initially, $v$ is the root) and, when a new letter $T[i]$ arrives, we put $v = \trans{u}{T[i]}$ for the first $u$ in the series $v, \link{v}, \link{\link{v}}, \ldots$ for which $\trans{u}{T[i]}$ is not $\mathsf{nil}$, or put $v = root$ if there is no such $u$, then we report all patterns ending at position $i$ using the report links and the flag $\flag{v}$ (see details in, e.g.,~\cite{CrochemoreRytter}).

\begin{figure}[htb]
%\scriptsize
\centering
\begin{subfigure}[t]{.55\textwidth}
\includegraphics[scale=0.34]{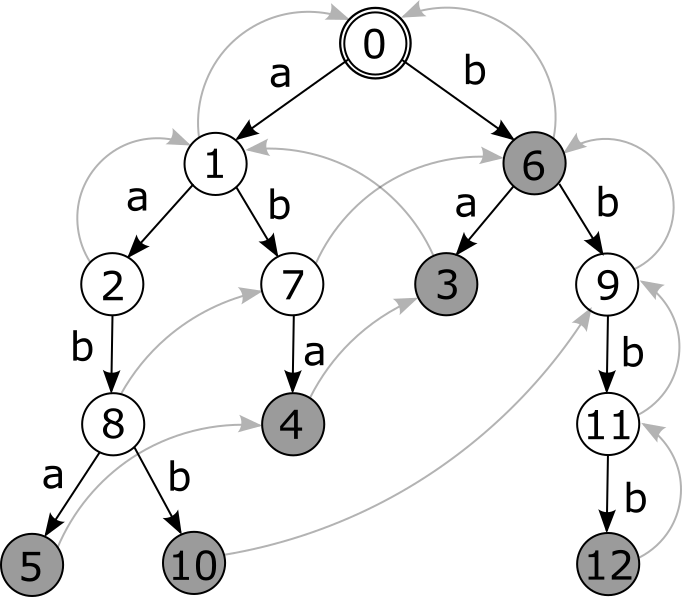}%,width=\linewidth
\caption{The trie $\tree{\dict}$ and failure links; each vertex $v$ is numbered by $\num{v}$ and is shaded iff $\flag{v} = 1$.}
\label{fig:automaton}
\end{subfigure}
\hfill
\begin{subtable}[t]{.42\textwidth}
%\vskip-64.0mm %for normalsize version
%\vskip-39.4mm %for scriptsize version
\vskip-56.6mm
\small
\setlength\tabcolsep{2.7pt}
\begin{tabular}{>{\small}r>{\small}r>{\small}c>{\small}c}
\num{v} & \str{v}  & \numdictgen{a} & \numdictgen{b} \\\cline{3-4}
0 &      & \multicolumn{1}{|>{\footnotesize}c}{1} & \multicolumn{1}{|>{\footnotesize}c|}{1} \\\cline{3-4}
1 & $a$    & \multicolumn{1}{|>{\footnotesize}c}{1} & \multicolumn{1}{|>{\footnotesize}c|}{1} \\\cline{3-4}
2 & $aa$   & \multicolumn{1}{|>{\footnotesize}c}{0} & \multicolumn{1}{|>{\footnotesize}c|}{1} \\\cline{3-4}
3 & $ba$   & \multicolumn{1}{|>{\footnotesize}c}{0} & \multicolumn{1}{|>{\footnotesize}c|}{0} \\\cline{3-4}
4 & $aba$  & \multicolumn{1}{|>{\footnotesize}c}{0} & \multicolumn{1}{|>{\footnotesize}c|}{0} \\\cline{3-4}
5 & $aaba$ & \multicolumn{1}{|>{\footnotesize}c}{0} & \multicolumn{1}{|>{\footnotesize}c|}{0} \\\cline{3-4}
6 & $b$    & \multicolumn{1}{|>{\footnotesize}c}{1} & \multicolumn{1}{|>{\footnotesize}c|}{1} \\\cline{3-4}
7 & $ab$   & \multicolumn{1}{|>{\footnotesize}c}{1} & \multicolumn{1}{|>{\footnotesize}c|}{0} \\\cline{3-4}
8 & $aab$  & \multicolumn{1}{|>{\footnotesize}c}{1} & \multicolumn{1}{|>{\footnotesize}c|}{1} \\\cline{3-4}
9 & $bb$   & \multicolumn{1}{|>{\footnotesize}c}{0} & \multicolumn{1}{|>{\footnotesize}c|}{1} \\\cline{3-4}
10 & $aabb$ & \multicolumn{1}{|>{\footnotesize}c}{0} & \multicolumn{1}{|>{\footnotesize}c|}{0} \\\cline{3-4}
11 & $bbb$  & \multicolumn{1}{|>{\footnotesize}c}{0} & \multicolumn{1}{|>{\footnotesize}c|}{1} \\\cline{3-4}
12 & $bbbb$ & \multicolumn{1}{|>{\footnotesize}c}{0} & \multicolumn{1}{|>{\footnotesize}c|}{0} \\\cline{3-4}
\end{tabular}
\caption{Encoding of transitions $\transdef$ in the trie $\tree{\dict}$ using the bit arrays $\numdictgen{a}$ and $\numdictgen{b}$.}
\label{fig:bitvec}
\end{subtable}
\caption{The trie $\tree{\dict}$ with $\dict = \{aaba, aabb, aba, b, ba, bbbb \}$.}
\label{fig:example}
\end{figure}

It is easy to see that the data structure occupies $O(m\log m)$ bits of space, where $m$ is the number of edges in $\tree{\dict}$. Let us describe now how Belazzougui could reduce the space consumption of this solution with no slowdown (see~\cite{Belazzougui} for a more detailed explanation).

First, he assigned to each vertex $v$ of $\tree{\dict}$ a unique number $\num{v} \in [0..m]$ so that, for any two vertices $u$ and $v$, we have $\num{u} < \num{v}$ iff $\str{u}^r < \str{v}^r$ lexicographically (see Figure~\ref{fig:automaton}). As it was pointed out by Hon et al.~\cite{HonEtAl}, this subtle numbering scheme corresponds to the numbering of vertices in the so-called XBW of $\tree{\dict}$ (see \cite{FerraginaEtAlXBW}), a generalization of the classical Burrows--Wheeler transform (BWT)~\cite{BurrowsWheeler} for tries. It turns out that the numbering allows us to organize fast navigation in the trie in small space, simulating the tables $\transdef$ in a manner that resembles the so-called FM-indexes~\cite{FerraginaManzini} based on the BWT.

For each letter $c \in [0..\sigma)$, define a bit array $\numdictgen{c}[0..m]$ such that, for any vertex $v$, we have $\numdictgen{c}[\num{v}] = 1$ iff $\trans{v}{c} \ne \mathsf{nil}$ (see Figure~\ref{fig:bitvec}). Let $\mathsf{rank}(i,\numdictgen{c})$ be an operation on $\numdictgen{c}$, called \emph{partial rank}, that returns $\mathsf{nil}$ if $\numdictgen{c}[i] \ne 1$, and returns the number of ones in $\numdictgen{c}[0..i]$ otherwise. By standard arguments, one can show that $\num{\trans{v}{c}} = \mathsf{rank}(\num{v},\numdictgen{c}) + e_{{<}c}$, provided $\trans{v}{c} \ne \mathsf{nil}$, where $e_{{<}c}$ is the number of edges in the trie with labels smaller than $c$ (see~\cite{Belazzougui,FerraginaEtAlXBW}). Let us concatenate $\numdictgen{0}, \numdictgen{1}, \ldots, \numdictgen{\sigma{-}1}$, thus obtaining a new bit array $\numdictdef$ of length $(m{+}1)\sigma$. Since $e_{{<}c}$ is equal to the number of ones in the arrays $\numdictgen{0}, \numdictgen{1}, \ldots, \numdictgen{c{-}1}$, we obtain $\num{\trans{v}{c}} = \mathsf{rank}(c(m{+}1) + \num{v}, \numdictdef)$ and $\mathsf{rank}(c(m{+}1) + \num{v}, \numdictdef) = \mathsf{nil}$ iff $\trans{v}{c} = \mathsf{nil}$. In order to support $\mathsf{rank}$ in $O(1)$ time, we equip $\numdictdef$ with the following data structure, which, in addition, supports the operation $\mathsf{select}(i,\numdictdef)$ that returns the position of the $i$th one in $\numdictdef$ ($\mathsf{select}$ is used below to compute the parent of $v$ by number $\num{v}$; see~\cite{Belazzougui}).

\begin{lemma}[see~\cite{RamanRamanRao}]
\label{lemma:rrr}
Every bit array of length $n$ with $m$ ones has an encoding that occupies $\log\binom{n}{m} + o(m)$ bits and supports select and partial rank queries in $O(1)$ time.
\end{lemma}

Since $\numdictdef$ contains exactly $m$ ones, the space occupied by $\numdictdef$ encoded as in Lemma~\ref{lemma:rrr} is $\log\binom{(m + 1)\sigma}{m} + o(m)$, which in~\cite{Belazzougui} was estimated by $m\log\sigma{+}1.443m{+}o(m)$ (here $1.443{\approx}\log e$). To encode $\flagdef$, Belazzougui constructs a bit array of length $m + 1$ containing exactly $d$ ones at positions $\num{v}$ for all $d$ vertices $v$ with $\flag{v} = 1$, and stores the array as in Lemma~\ref{lemma:rrr} (for this array, we need only access queries, which can be simulated by $\mathsf{rank}$), thus occupying $\log\binom{m+1}{d}{+}o(d) \le O(d\log\frac{m}{d})$ bits of space. It remains to encode failure and report links.

Belazzougui noticed that the failure links form a tree on the vertices of $\tree{\dict}$ in which, for vertex $v$, $\link{v}$ returns the parent of $v$; more importantly, the numbering $\numdef$ corresponds to the order of vertices in a depth first traversal of this tree. This allows us to represent the tree of failure links in $2m + o(m)$ bits using the following lemma.

\begin{lemma}[{see \cite{NavarroSadakane}}]
\label{lemma:succinct_tree}
Every tree with $m$ edges can be encoded in $2m + o(m)$ bits with the support of the operation $\mathsf{parent}(v)$, which returns the parent of vertex $v$, in constant time, provided all the vertices are represented by their numbers in a depth first traversal.
\end{lemma}

An analogous observation is also true for the report links and the tree induced by them has at most $d$ internal vertices, which allows to spend only $O(d\log\frac{m}{d}) + o(m)$ bits for it.

\begin{lemma}[{see \cite{JanssonSadakaneSung}}]
\label{lemma:ultra_succinct_tree}
Every tree with $m$ edges and $d$ internal vertices can be encoded in $d\log\frac{m}{d} + O(d) + o(m)$ bits with the support of the operation $\mathsf{parent}(v)$ in constant time, provided all the vertices are represented by their numbers in a depth first traversal.
\end{lemma}

Thus, the total space consumed by the succinct versions of the structures $\flagdef$, $\transdef$, $\linkdef$, and $\reportdef$ is $m\log\sigma + 3.443 m + o(m) + O(d\log\frac{m}{d})$ bits.

\section{Compression Boosting}
\label{sec:compression_boosting}

In~\cite{Belazzougui} it was noticed that if each array $\numdictgen{0}, \numdictgen{1}, \ldots, \numdictgen{\sigma{-}1}$ is encoded separately using Lemma~\ref{lemma:rrr}, then they altogether occupy $\sum_{0 \le c < \sigma} (\log\binom{m + 1}{n_c} + o(n_c))$ bits, where $n_c$ is the number of labels $c$ in $\tree{\dict}$, which is upper bounded by $\sum_{0 \le c < \sigma} (n_c\log\frac{m}{n_c} + 1.443 n_c) + o(m) = mH_0(\tree{\dict}) + 1.443m + o(m)$ (here we apply the inequalities $\log\binom{m + 1}{n} \le n\log\frac{m + 1}{n} + n\log e$ and $\log e < 1.443$, and estimate $n_c\log\frac{m + 1}{n_c}$ as $n_c\log\frac{m}{n_c} + o(n_c)$; see~\cite{Belazzougui}). Using such separated array encodings and some auxiliary data structures, one can reduce the space of the whole data structure to $m H_0(\tree{\dict}) + 3.443 m + o(m) + O(d\log\frac{m}{d})$ bits, provided $\sigma \le m^\delta$ for some constant $\delta < 1$. Hon et al.~\cite{HonEtAl} further developed this idea, applying the compression boosting technique~\cite{FerraginaEtAlBoosting}.

\subparagraph{Compression boosting}
Choose an integer $k$ such that $0 \le k \le \alpha\log_\sigma m-1$, where $\alpha \in (0,1)$ is an arbitrary fixed constant. For $i \in [0..m]$, denote by $\istr{i}$ the string $\str{v}$ such that $v$ is the vertex of $\tree{\dict}$ with $\num{v} = i$. For technical reasons, we introduce a special letter $\$ = \sigma$. Hon et al.~\cite{HonEtAl} partition the set $[0..m]$ into segments in which the strings $\istr{i}$, for $i$ from the same segment, have a common suffix of length $k$; then the subarrays of $\numdictgen{c}$ corresponding to these segments are encoded separately using Lemma~\ref{lemma:rrr}. More precisely, let $[\ell_\rho .. r_\rho]$, for $\rho \in [0..\sigma]^k$, be the set of all $i \in [0..m]$ such that $\rho$ is a suffix of $\$^k\istr{i}$ (the definition of $\numdef$ implies that this set forms a segment); each array $\numdictgen{c}[0..m]$, for $c \in [0..\sigma)$, is partitioned into the subarrays $\numdictgen{c}[\ell_{\rho} .. r_{\rho}]$, where $\rho \in [0..\sigma]^k$ (empty segments $[\ell_{\rho} .. r_{\rho}]$ are excluded), and each such subarray is encoded separately using Lemma~\ref{lemma:rrr}. Then, all the encoded subarrays in total occupy $\sum_{0\le c<\sigma} \sum_{\rho \in [0..\sigma]^k} (\log\binom{m_\rho}{n_{c,\rho}} + o(n_{c,\rho}))$ bits, where $m_\rho = r_\rho - \ell_\rho + 1$ and $n_{c,\rho}$ is the number of ones in the subarray $\numdictgen{c}[\ell_\rho .. r_\rho]$. Hon et al.~upper bound this sum by $m H_k^*(\tree{\dict}) + 1.443m + o(m)$, where $H_k^*(\tree{\dict}) = \frac{m + \ell}{m + 1}H_k(\treedef^*)$ is their definition of the $k$th-order empirical entropy (see~\cite{HonEtAl}; for simplicity, we use our notation $H_k$ to define $H_k^*$) in which $\ell$ is the number of leaves in $\tree{\dict}$ and $\treedef^*$ is the trie obtained from $\tree{\dict}$ by attaching to each leaf an outgoing edge labeled with $\$$. Note that $H_k^*(\tree{\dict}) \ge H_k(\treedef^*)$.

We believe that the definition of $H_k^*$ by Hon et al.~is not satisfactory. The problem is that the inequality $H_k^*(\tree{\dict}) \le \log\sigma$ (and even $H_k(\treedef^*) \le \log\sigma$), which seems to be natural for any proper definition of the empirical entropy, does not necessarily hold; as a corollary, according to the analysis of Hon et al., the encoding of $\transdef$ occupying $m\log\sigma + 1.443m$ bits in the Belazzougui's data structure might ``grow'' after compression up to $m H_k^*(\tree{\dict}) + 1.443m$ bits (however, there is no growth in reality, just the upper bound of Hon et al.~is too rough). For instance, one can observe such behavior in the trie $\treedef$ of all strings of length $h$ over the alphabet $\{0,1\}$: while it is straightforward that $H_1(\treedef) = \log\sigma = 1$, it can be shown that $H_1(\treedef^*) \approx \log 3$ since, for $c \in \{0,1\}$, the string $\treedef^*_c$ consisting of all labels in $\treedef^*$ with ``context'' $c$ contains roughly $m / 4$ of each of the letters $0$, $1$, and $\$$, where $m = 2^{h+1} - 2$ is the number of edges in $\treedef$ (we omit further details as they are straightforward).

For brevity, let us denote the summations $\sum_{\rho \in [0..\sigma]^k}$ and $\sum_{0 \le c < \sigma}$ in this paragraph by $\sum_\rho$ and $\sum_c$, respectively. For the sake of completeness, we show that the compression boosting technique allows to achieve the $k$th-order empirical entropy $H_k(\tree{\dict})$, which, unlike $H_k^*$, satisfies the inequality $\log\sigma \ge H_k(\tree{\dict})$, i.e., we are to prove that $\sum_c \sum_\rho (\log\binom{m_\rho}{n_{c,\rho}} + o(n_{c,\rho})) \le mH_k(\tree{\dict}) + 1.443m + o(m)$ (but we do not discuss additional structures of Hon et al. required for navigation; see~\cite{HonEtAl}). First, $\sum_c \sum_\rho (\log\binom{m_\rho}{n_{c,\rho}} + o(n_{c,\rho}))$ is upper bounded by $\sum_\rho \sum_c (n_{c,\rho}\log\frac{m_\rho}{n_{c,\rho}} + 1.443n_{c,\rho}) + o(m) = \sum_\rho \sum_c n_{c,\rho}\log\frac{m_\rho}{n_{c,\rho}} + 1.443m + o(m)$. Denote $n_\rho = \sum_c n_{c,\rho}$. Note that, by definition, we have $m H_k(\tree{\dict}) = \sum_\rho \sum_c n_{c,\rho}\log\frac{n_\rho}{n_{c,\rho}}$. Since the function $\log x$ is concave, we have $\log(x + d) \le \log x + d(\log x)' = \log x + d\log e / x$, for any $x > 0$ and any real $d$ such that $x + d > 0$. Hence, we deduce $\sum_\rho \sum_c n_{c,\rho}\log\frac{m_\rho}{n_{c,\rho}} = \sum_\rho \sum_c n_{c,\rho}\log(\frac{n_\rho}{n_{c,\rho}} + \frac{m_\rho - n_\rho}{n_{c,\rho}}) \le \sum_\rho \sum_c n_{c,\rho}(\log\frac{n_\rho}{n_{c,\rho}} + \frac{m_\rho - n_\rho}{n_{\rho}} \log e) = \sum_\rho \sum_c n_{c,\rho}\log\frac{n_\rho}{n_{c,\rho}} + \sum_\rho (m_\rho - n_\rho) \log e = m H_k(\tree{\dict}) + \log e$; the equality $\sum_\rho (m_\rho - n_\rho) \log e = \log e$ holds since $\sum_\rho m_\rho = m + 1$ and $\sum_\rho n_\rho = m$. Finally, we hide the constant $\log e$ under $o(m)$ and obtain the $k$th-order entropy compression: $\sum_c \sum_\rho (\log\binom{m_\rho}{n_{c,\rho}} + o(n_{c,\rho})) \le m H_k(\tree{\dict}) + 1.443m + o(m)$.

\subparagraph{Fixed block compression boosting}
The described partition lacks uniformity and requires relatively complex auxiliary data structures in order to support navigation and queries. Hon et al.~\cite{HonEtAl} indeed organize such an infrastructure using $o(m)$ bits, provided $\sigma \le m^\delta$ for a constant $\delta < 1$ (the condition $0 \le k \le \alpha\log_\sigma m - 1$ plays its role in this part). But it turns out that we can considerably simplify their whole construction using the fixed block boosting by K{\"a}rkk{\"a}inen and Puglisi~\cite{KarkkainenPuglisi}. The approach relies on the following lemma.

\begin{lemma}[{see~\cite[Lemma 4]{KarkkainenPuglisi}}]
\label{lemma:fixed_length}
Let $s = s_1 s_2 \cdots s_{\ell}$ be an arbitrary partition of a string $s$ into $\ell$ substrings. Let $s = s'_1 s'_2 \cdots s'_t$ be a different partition of $s$ into $t$ substrings each of which has length at most $b$. Then, we have
$\sum_{i=1}^t |s'_i| H_0(s'_i) \le \sum_{i=1}^\ell |s_i|H_0(s_i) + (\ell - 1) b$.
\end{lemma}

We encode the bit array $\numdictdef$, which represents the transitions $\transdef$, as follows.

\begin{lemma}
\label{lemma:compression_boosting}
Provided $\sigma \le m^\delta$ for a constant $\delta < 1$, $\numdictdef$ has an encoding that supports select and partial rank in $O(1)$ time and occupies $mH_k(\tree{\dict}) + 1.443m + o(m)$ bits si\-mul\-ta\-ne\-ous\-ly for all $k \in [0..\max\{0,\alpha\log_\sigma m{-}2\}]$, where $\alpha \in (0,1)$ is an arbitrary fixed constant.
\end{lemma}
\begin{proof}
We first discuss a fixed block encoding of the arrays $\numdictgen{c}$ and prove, by means of Lemma~\ref{lemma:fixed_length}, that it achieves the $k$th-order entropy compression. Then, we describe auxiliary structures that occupy only $o(m)$ bits and are used for queries and navigation in the blocks.

We partition each array $\numdictgen{c}[0..m]$ into $t = \lceil\frac{m + 1}{b}\rceil$ blocks of length $b = \sigma\lceil\log^2 m\rceil$ (the last block can be shorter), encode each block using Lemma~\ref{lemma:rrr}, and concatenate the encodings. Thus, we consume $\sum_{i=1}^t \sum_{0 \le c < \sigma} (\log\binom{b_i}{n_{c,i}} + o(n_{c,i}))$ bits, where $n_{c,i}$ is the number of ones in the block $\numdictgen{c}[(i{-}1)b .. \min\{ib{-}1,m\}]$ and $b_i$ is the length of the $i$th block (so that $b_i = b$, for $i \in [0..t)$, and $b_t \le b$). Since $\sum_{i=1}^t \log\binom{b_i}{n_{c,i}} \le \log\binom{m+1}{n_c}$ for $n_c = \sum_{i=1}^t n_{c,i}$, the result trivially holds for $k = 0$. The sum is upper bounded by $\sum_{i=1}^t \sum_{0 \le c < \sigma} n_{c,i}\log\frac{b_i}{n_{c,i}} + 1.443 m + o(m)$. Let us prove that $\sum_{i=1}^t \sum_{0 \le c < \sigma} n_{c,i}\log\frac{b_i}{n_{c,i}} \le m H_k(\tree{\dict}) + o(m)$ for all $k \in (0..\alpha\log_\sigma m{-}2]$.

Fix $k \in (0..\alpha\log_\sigma m{-}2]$. Denote $n_i = \sum_{0 \le c < \sigma} n_{c,i}$. Using the inequality $\log(x + d) \le \log x + d \log e / x$, we deduce the following upper bound:

$$
\sum_{\substack{1 \le i \le t\\0 \le c < \sigma}} n_{c,i}\log\frac{b_i}{n_{c,i}} 
= \sum_{\substack{1 \le i \le t\\0 \le c < \sigma}} n_{c,i}\log(\frac{n_i}{n_{c,i}} + \frac{b_i - n_i}{n_{c,i}}) \le \sum_{\substack{1 \le i \le t\\0 \le c < \sigma}} n_{c,i}\log\frac{n_i}{n_{c,i}} + \sum_{i=1}^t (b_i - n_i) \log e.$$

First, we have $\sum_{i=1}^t (b_i - n_i) \log e = \log e$ since $\sum_{i=1}^t b_i = m + 1$ and $\sum_{i=1}^t n_i = m$. Second, $\sum_{i=1}^t \sum_{0 \le c < \sigma} n_{c,i}\log\frac{n_i}{n_{c,i}} = \sum_{i=1}^t|s'_i|H_0(s'_i)$, where $s'_i$ is a string of length $n_i$ formed by concatenating the letters on the edges $(u,v)$ such that $\num{u}$ is inside the $i$th block, i.e., $\num{u} \in [(i{-}1)b~..~(i{-}1)b + b_i)$. For $\rho \in [0..\sigma]^k$, let $[\ell_\rho .. r_\rho]$ be the set of all $i$ such that $\rho$ is a suffix of $\$^k\istr{i}$, and let $s_\rho$ be a string formed by concatenating the letters on the edges $(u,v)$ such that $\num{u} \in [\ell_\rho .. r_\rho]$. By definition, $\sum_{\rho \in [0..\sigma]^k} |s_\rho| H_0(s_\rho) = m H_k(\tree{\dict})$. Note that at most $\ell = \sum_{i=0}^k \sigma^i = \frac{\sigma}{\sigma - 1}(\sigma^k - \frac{1}{\sigma})$ strings $s_\rho$ are nonempty and $\ell \le 2\sigma^k$ since $\sigma \ge 2$. Let $\rho_1, \rho_2, \ldots, \rho_{(\sigma + 1)^k}$ be an ordering of all strings $\rho \in [0..\sigma]^k$ such that $\ell_{\rho_1} \le \cdots \le \ell_{\rho_{(\sigma + 1)^k}}$. The definitions of $s'_i$ and $s_\rho$ imply that the letters in $s'_i$ and $s_\rho$ can be arranged so that $s'_1 s'_2 \cdots s'_t = s_{\rho_1} s_{\rho_2} \cdots s_{\rho_{(\sigma + 1)^k}}$. Therefore, by Lemma~\ref{lemma:fixed_length}, we obtain the next inequality:

$$\sum_{i=1}^t |s'_i| H_0(s'_i) \le \sum_{\rho \in [0..\sigma]^k} |s_\rho| H_0(s_\rho) + (\ell - 1) \max_{i \in [1..t]} n_i \le m H_k(\tree{\dict}) + 2\sigma^{k+1} b.$$

As $k \le \alpha\log_\sigma m-2$, we have $\sigma^{k+1}b = \sigma^{k+2}\lceil\log^2 m\rceil \le m^{\alpha}\lceil\log^2 m\rceil = o(m)$.

It remains to describe the auxiliary data structures that help to answer select and partial rank queries on the (now virtual) array $\numdictdef$. First, we store $\sigma t$ pointers to the data structures encoding the blocks $\numdictgen{c}[(i{-}1)b .. \min\{ib{-}1, m\}]$, for $c \in [0..\sigma)$ and $i \in [1..t]$. For rank, we create an array of length $\sigma t$ that stores the number of ones in the subarrays $\numdictdef[0..(m + 1)c + ib - 1]$, for $i \in [0..t)$ and $c \in [0..\sigma)$. All this takes $O(\sigma t \log m) = O(\frac{m}{\log^2 m} \log m) = o(m)$ bits. For select, we create a bit array $S$ formed by concatenating unary encodings for the number of ones in the blocks: e.g., if the first four blocks (of all $\sigma t$ blocks) contain, resp., $3$, $2$, $0$, and $2$ ones, then $S = 11101100110\cdots$; $S$ is encoded using Lemma~\ref{lemma:rrr} and, thus, occupies $\log\binom{m + \sigma t}{m} + o(m) = \log\binom{m + \sigma t}{\sigma t} + o(m) \le O(\sigma t\log m) + o(m) = o(m)$ bits. Using these structures, one can straightforwardly perform select and partial rank on $\numdictdef$ in $O(1)$ time.
\end{proof}

\section{Main Data Structure}
\label{sec:space_reduction}

The encoding of failure links imposes a $2m$-bit overhead, which, for small alphabet or highly compressible data, might be comparable to the space required for other structures. In this section, we describe a different encoding that uses only $\varepsilon m + o(m)$ bits, for any $\varepsilon \in (0, 2)$.

The key idea is to store the failure links only for some trie vertices. We call a subset $W$ of the vertices of a tree a \emph{$t$-dense} subset if each vertex $v \notin W$ has an ancestor $p \in W$ located at a distance less than $t$ edges from $v$. (Note that the definition implies that $W$ must contain the root.) Now we can formulate the main lemma.

\begin{lemma}
\label{lemma:kdense_aho}
Suppose that $\link{v}$ can be calculated in $O(1)$ time only for $v \in W$, where $W$ is a vertex set that is $t$-dense in the tree $\tree{\dict}$; then there is a modification of the Aho--Corasick algorithm that uses the links $\link{v}$ only for $v \in W$ and processes any text $T$ in $O(t|T| + occ)$ time, where $occ$ is the number of occurrences of the patterns in $T$.
\end{lemma}
\begin{proof}
Our algorithm essentially simulates the Aho--Corasick solution but in the case when the usual algorithm calculates $\link{v}$ for $v \notin W$, the new one instead finds the nearest ancestor $p \in W$ of $v$, ``backtracks'' the input string $T$ accordingly, then computes $\link{p}$, and continues the execution from this point (if $p$ is the root and, thus, $\link{p}$ is undefined, we simply skip one letter and continue). The pseudocode is as follows (the omitted code reporting patterns in line~\ref{lst:report} simply traverses report links and checks whether $\flag{v} = 1$):

\algtext*{EndFunction}% Remove "end function" text
\algtext*{EndProcedure}% Remove "end procedure" text
\algtext*{EndFor}% Remove "end for" text
\algtext*{EndWhile}% Remove "end while" text
\algtext*{EndIf}% Remove "end if" text
\algloopdefx{NIf}[1]{\textbf{if} #1 \textbf{then}}
\algloopdefx{NElse}{\textbf{else}}
\algloopdefx{NElseIf}{\textbf{else if}}
\algloopdefx{NForAll}[1]{\textbf{for each} #1 \textbf{do}}
\algloopdefx{NWhile}[1]{\textbf{while} #1 \textbf{do}}
\algloopdefx{NFor}[1]{\textbf{for} #1 \textbf{do}}

%\algsetup{linenosize=\tiny}
\begin{algorithmic}[1]
%\scriptsize
\Function{$auto$}{$v$, $i$, $i_{max}$}
    \If{$i > i_{max}$}
        $i_{max} \gets i$ and report all patterns ending at position $i - 1$;\label{lst:report}
    \EndIf
    \If{$i = |T|$}
        \Return{$v$};
    \EndIf
    \If{$\trans{v}{T[i]} \neq \mathbf{nil}$}
    	\Return{$auto(\trans{v}{T[i]}, i + 1, i_{max})$};\label{lst:exit_on_go}
    \EndIf
    \For{$(p \gets v;\; p \notin W;\; i \gets i - 1)$}\label{lst:loop_p_beg} \Comment{C-style loop with three parameters}
    	\State $(p, c) \gets \mathsf{parent\_edge}(v)$;    \Comment{$p$ and $c$ are such that $v = \trans{p}{c}$}
    \EndFor\label{lst:loop_p_end}
    \If{$p$ is $root$}
        \Return{$auto(root, i + 1, i_{max})$};\label{lst:root_recurse}
    \EndIf
    \State\Return{$auto(\link{p}, i, i_{max})$};\label{lst:failure_recurse}
\EndFunction
\end{algorithmic}

The execution starts with $auto(root, 0, 0)$. The function $\mathsf{parent\_edge}(v)$ returns the parent $p$ of $v$ and the letter $c$ such that $v = \trans{p}{c}$ (note that we only use $p$); in~\cite{Belazzougui} it was shown that $\num{p} = x \bmod (m + 1)$ and $c = \lfloor x / (m + 1)\rfloor$, where $x = \mathsf{select}(\num{v}, \numdictdef)$ and $\numdictdef$ is the bit array that encodes the $\transdef$ transitions (see above). To prove the correctness, let us show by induction on the length of $T$ that $auto(root, 0, 0)$ returns the vertex $v_T$ such that $\str{v_T}$ is the longest suffix of $T$ that can be spelled out on a root-vertex path in $\tree{\dict}$.

The base $|T| = 0$ is trivial. Suppose that the claim holds for all lengths smaller than $|T|$. We are to show that $auto(root, 0, 0) = v_T$. By the inductive hypothesis, when $auto(v, i, i_{max})$ is called with $i = |T|{-}1$ for the first time, the string $\str{v}$ is the longest suffix of $T[0..|T|{-}2]$ that can be read on a root-vertex path of $\tree{\dict}$. Therefore, if $\trans{v}{T[|T|{-}1]} \ne \mathbf{nil}$, $\trans{v}{T[|T|{-}1]}$ obviously is equal to $v_T$ and we return it in line~\ref{lst:exit_on_go}. Now suppose that $\trans{v}{T[|T|{-}1]} = \mathbf{nil}$. In lines~\ref{lst:loop_p_beg}--\ref{lst:loop_p_end} we find the nearest ancestor $p$ of $v$ belonging to $W$ or put $p = v$ if $v \in W$, and backtrack accordingly to $T[0..i]$, for $i \in [0..|T|)$, such that $\str{p}$ is a suffix of $T[0..i{-}1]$. Observe the following claims: (i)~for any vertex $w$ such that $\str{w}$ is a suffix of $T[0..i{-}1]$, the result of $auto(w, i, i_{max})$ is the same as the result of $auto(root, 0, 0)$ with $T := T[i{-}|\str{w}|\ ..\ |T|{-}1]$; (ii)~$\str{v_T}$ is either the empty string or a proper suffix of $\str{v}$ concatenated with $T[|T|{-}1]$. When $p$ is not the root, the claim (ii) and the fact that $\str{\link{p}}$ is the longest proper suffix of $\str{p}$ present in the trie imply that $\str{v_T}$ is a suffix of $T[i{-}|\str{\link{p}}|\ ..\ |T|{-}1]$. Then, by the claim (i) and the inductive hypothesis, the recursion $auto(\link{p}, i, i_{max})$ in line~\ref{lst:failure_recurse} returns $v_T$. When $p$ is the root, (i) and (ii) analogously imply that the call to $auto(root, i + 1, i_{max})$ in line~\ref{lst:root_recurse} returns $v_T$.

Let us analyze the time complexity. The algorithm maintains two indices: $i$ and $k = i - |\str{v}|$. Each call to $auto(v, i, i_{max})$ with $i \ne |T|$ either increases $i$ in line~\ref{lst:exit_on_go} or increases $k$ in lines~\ref{lst:root_recurse} or~\ref{lst:failure_recurse}. Since $W$ is a $t$-dense subset, the loop~\ref{lst:loop_p_beg}--\ref{lst:loop_p_end} can decrease $i$ by at most $t$ before increasing $k$. Therefore, $i$ can be decreased by at most $t|T|$ in total and, hence, the running time of the whole algorithm is $O(t|T|)$ plus $O(occ)$ time to report pattern occurrences.

The presented solution explicitly stores $T$ (or at least its last $m$ letters) in order to support ``backtracking'' during the calculations. We, however, cannot afford to allocate the $m\log\sigma$ bits for $T$ and desire to fit the additional space within an $o(m)$ bound. To this end, we maintain only a substring $T[i..i']$ such that $i' - i \le 2\sqrt{m}$. While decreasing $i$ in the loop~\ref{lst:loop_p_beg}--\ref{lst:loop_p_end}, we grow this substring to the left using the letters $c$ returned by $\mathsf{parent\_edge}(v)$. Once $i' - i$ becomes larger than $2\sqrt{m}$, we simply decrement $i'$. Once $i$ becomes larger than $i'$ and $i' < i_{max}$, we must somehow restore the letters $T[i], T[i{+}1], \ldots$  Denote by $P$ the set of all positions $j \in [0..i_{max}]$ such that $j$ is a multiple of $\lceil\sqrt{m}\rceil$. For each $j \in P \cup \{i_{max}\}$, we store a vertex $v_j$ that was the vertex $v$ in the function $auto$ when we reached the position $j$ for the first time (so that $\str{v_j}$ is a suffix of $T[0..j{-}1]$). Since the loop~\ref{lst:loop_p_beg}--\ref{lst:loop_p_end} cannot make $i$ smaller than $i - |\str{v}|$, one can easily show that we always have $i \ge j - |\str{v_j}|$ for each $j \in P \cup \{i_{max}\}$. Thus, once $i > i'$, we compute in $O(1)$ time the position $j = \min\{\{i_{max}\} \cup \{j \in P \colon j \ge i + \sqrt{m}\}\}$, then put $i' = j - 1$, and restore the string $T[i..i']$ in $O(\sqrt{m})$ time iteratively applying the function $\mathsf{parent\_edge}$ to the vertex $v_j$.

Since $i_{max} - i$ cannot exceed $m$, it suffices to store $v_j$ only for the $\lceil\sqrt{m}\rceil + 1$ largest positions from $P$. One can maintain these $v_j$ in a straightforward way using a deque on circular array of length $O(\sqrt{m})$, so that access to arbitrary $v_j$ takes $O(1)$ time. Therefore, the additional space used is $O(\sqrt{m}\log m) = o(m)$ bits. By standard arguments, one can show that the time $O(\sqrt{m})$ required to restore $T[i..j{-}1]$ is amortized among at least $\sqrt{m}$ increments of $i$ that were performed to make $i > i'$. Thus, the total running time is $O(t|T|)$ as in the version that stores $T$ explicitly.
\end{proof}

To perform $\link{v}$, for $v \in W$, and to check whether $v \in W$, we use the next lemma.

\begin{lemma}
\label{lemma:tree_compression}
Let $W$ be a subset of vertices of a rooted tree $\mathcal{F}$ with $m$ edges. There is an encoding of $\mathcal{F}$ that occupies $2|W|\log\frac{m+1}{|W|} + O(|W|) + o(m)$ bits and, for any vertex $v$, allows to determine whether $v \in W$ in $O(1)$ time and to compute the parent of $v$ if $v \in W$ in $O(1)$ time, provided all the vertices are represented by their numbers in a depth first~traversal~of~$\mathcal{F}$.
\end{lemma}
\begin{proof}
Let $\num{v}$ be a vertex numbering that corresponds to a depth first traversal of $\mathcal{F}$ (the numbers are from the range $[0..m]$). To check whether $v \in W$, we create a bit array $A$ of length $m + 1$ such that, for each vertex $v$, we have $A[\num{v}] = 1$ iff $v \in W$. The array $A$ is encoded in $\log\binom{m+1}{|W|} + o(|W|) = |W|\log\frac{m+1}{|W|} + O(|W|)$ bits in the data structure of Lemma~\ref{lemma:rrr} that supports  access to $A[i]$ in $O(1)$ time using partial rank queries.

\begin{figure}[htb]
\centering
\begin{subfigure}[t]{.35\textwidth}
\centering
\includegraphics[scale=0.2]{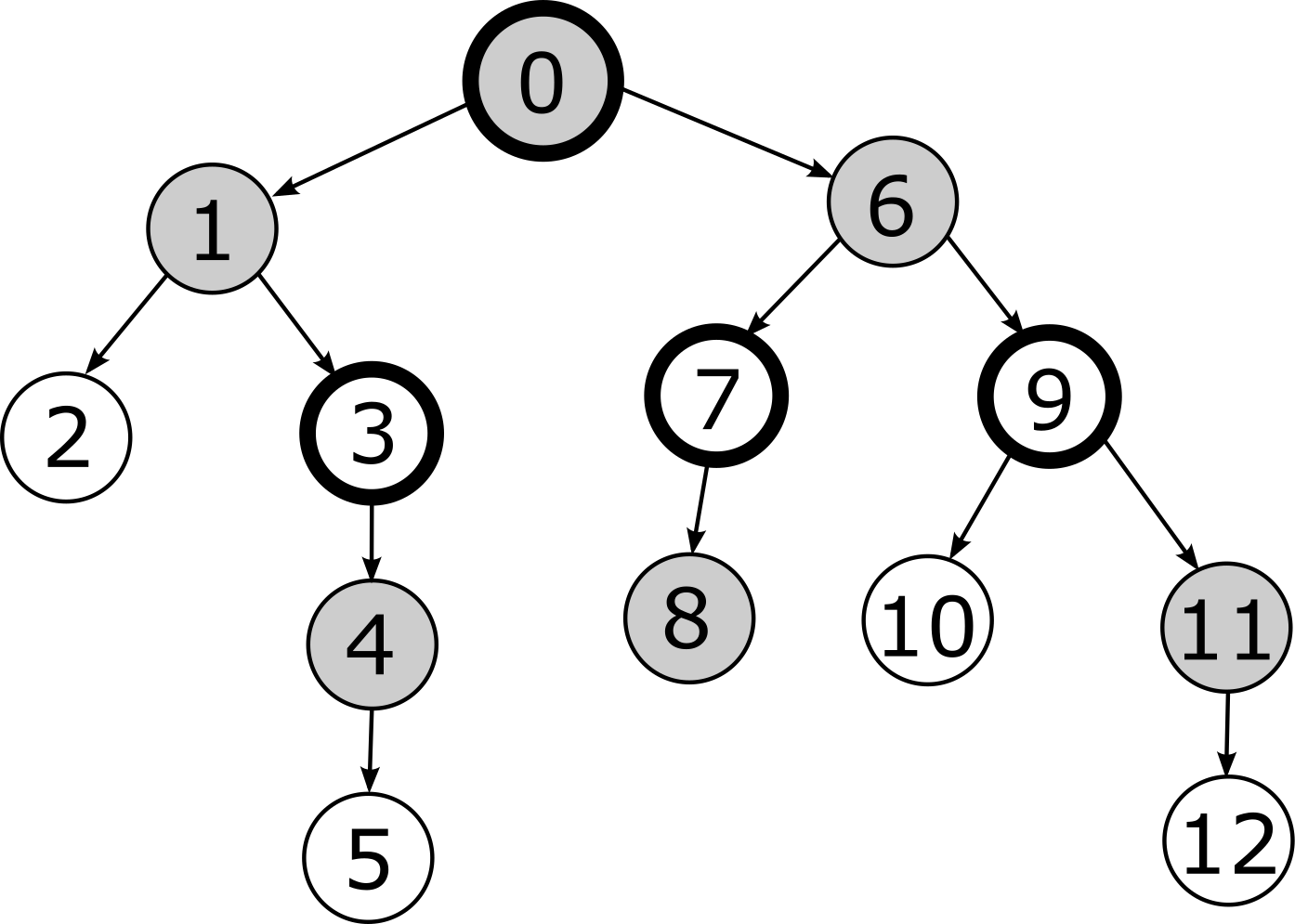}
\end{subfigure}
\hfill
\begin{subfigure}[t]{.6\textwidth}
\centering
\includegraphics[scale=0.2]{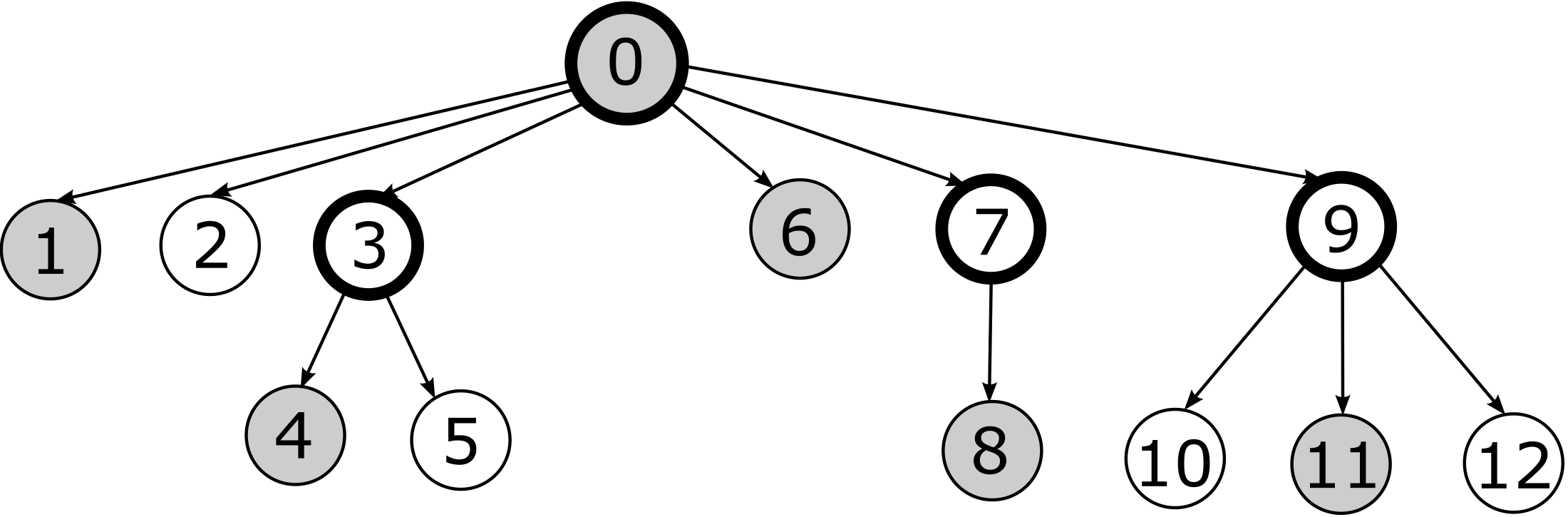}
\end{subfigure}
\caption{The transformation of Lemma~\ref{lemma:tree_compression} for the tree of failure links from Figure~\ref{fig:example}: the left tree is original, the right tree is transformed; vertices from $W$ are gray, important vertices are bold.}
\label{fig:trans}
\end{figure}

Let us now describe a transformation of the tree $\mathcal{F}$ that preserves the function $\textsf{parent}(v)$ for $v \in W$ (see Figure~\ref{fig:trans}). We call a vertex $v$ \emph{important} if either $v$ is the root or $v$ has a child from $W$. In the transformed tree, for each vertex $u$, the parent of $u$ is the nearest important ancestor of $u$ in $\mathcal{F}$. Obviously, there are at most $|W| + 1$ internal vertices in thus defined tree. Further, it is easy to see that the vertex numbering $\numdef$ corresponds to a depth first traversal of the transformed tree. Therefore, we can encode the new tree in $|W|\log\frac{m+1}{|W|} + O(|W|) + o(m)$ bits using the data structure of Lemma~\ref{lemma:ultra_succinct_tree}; for $v \in W$, the query $\mathsf{parent}(v)$ on this structure returns the number $\numdef$ of the parent of $v$ in the tree $\mathcal{F}$.
\end{proof}

Combining Lemmas~\ref{lemma:kdense_aho} and~\ref{lemma:tree_compression}, we prove the main theorem.

\begin{theorem}
\label{theorem:compact_tree}
Let $\varepsilon \in (0,2)$ be an arbitrary constant and let $\dict$ be a set of $d$ patterns over the alphabet $[0..\sigma)$ such that $\sigma \le m^\delta$, for some constant $\delta < 1$, where $m$ is the number of edges in the trie $\tree{\dict}$ containing $\dict$. Then, there is a data structure that allows to find all $occ$ occurrences of the patterns in any text $T$ in $O(|T| + occ)$ time and occupies $m H_k(\tree{\dict}) + 1.443m + \varepsilon m + O(d\log\frac{m}{d})$ bits simultaneously for all $k \in [0..\max\{0,\alpha\log_\sigma m{-}2\}]$, where $\alpha \in (0,1)$ is an arbitrary constant.
%For any constant $\varepsilon \in (0, 2)$ and any set $\dict$ of $d$ patterns over the alphabet $[0..\sigma)$, there is a data structure that allows to find all $occ$ occurrences of the patterns in any text $T$ in $O(|T| + occ)$ time and occupies $m H_k(\tree{\dict}) + 1.443m + \varepsilon m + O(d\log\frac{m}{d})$ bits of space simultaneously for all $k \in [0..\max\{0,\alpha\log_\sigma m{-}2\}]$, where $\alpha \in (0,1)$ is an arbitrary fixed constant, $\tree{\dict}$ is the trie containing $\dict$, and $m$ is the number of edges in $\tree{\dict}$.
\end{theorem}
\begin{proof}
Let us construct a small $t$-dense subset for the tree $\tree{\dict}$. Observe that the set $W_i$ that consists of the root and all vertices of $\tree{\dict}$ with height $h \equiv i \pmod{t}$ is a $t$-dense subset. Obviously, there exists $j \in [0..t)$ such that $|W_j| \le \lceil\frac{m+1}{t}\rceil$. We apply Lemma~\ref{lemma:tree_compression} to the tree $\mathcal{F}$ of failure links and the subset $W_j$, provided all vertices are represented by the numbering $\numdef$ defined in Section~\ref{sec:basic_algorithm} (in~\cite{Belazzougui} it was shown that this numbering corresponds to a depth first traversal of $\mathcal{F}$). By Lemma~\ref{lemma:kdense_aho}, this allows us to solve the dictionary matching problem in $O(t|T| + occ)$ time with only $2\lceil\frac{m+1}{t}\rceil\log{t} + O(\frac{m}{t}) + o(m)$ bits used for the failure links, which is upper bounded by $m \frac{c\log t}{t}$, for an appropriate constant $c > 0$. We add to the failure links the data structures for $\flagdef$, $\transdef$, and $\reportdef$ described in Sections~\ref{sec:basic_algorithm} and~\ref{sec:compression_boosting} (Lemma~\ref{lemma:compression_boosting}), which consume $m H_k(\tree{\dict}) + 1.443m + o(m) + O(d\log\frac{m}{d})$ bits (simultaneously for all $k \in [0..\max\{0,\alpha\log_\sigma m{-}2\}]$), and choose $t$ in such a way that $\frac{c\log t}{t} \le \varepsilon / 2$, so that the failure links take only $\varepsilon m / 2$ bits. Solving the equation, we obtain $t = \Theta(\varepsilon^{-1} \log{\varepsilon^{-1}})$. Since $\varepsilon$ is constant, the additive terms $o(m)$ in space can be upper bounded by $\varepsilon m / 2$ and $t$ in $O(t|T| + occ)$ can be hidden under the big-O, so that the total space is $m H_k(\tree{\dict}) + 1.443m + \varepsilon m + O(d\log\frac{m}{d})$ bits and the processing time is $O(|T| + occ)$.
\end{proof}

Note that, as it follows from the proof, the big-O notation in Theorem~\ref{theorem:compact_tree} hides a slowdown of the processing time to $O(|T|\varepsilon^{-1} \log{\varepsilon^{-1}} + occ)$, which is the price of the space improvement.

Let us now show that our solution is, in a sense, close to space optimal when $d = o(m)$.
It is easy to see that any data structure solving the multiple pattern matching problem implicitly encodes the trie $\tree{\dict}$ containing the dictionary $\dict$ of patterns: if the data structure is a black box, then we can enumerate all possible strings of length ${\le}m$ over the alphabet $[0..\sigma)$ and check which of them are recognized by the black box, thus finding all the patterns from $\dict$. It is known that the number of tries with $m$ edges over an alphabet of size $\sigma$ is at least $\frac{1}{\sigma(m + 1) + 1}\binom{\sigma(m + 1)}{m + 1}$ (see \cite[eq.~7.66]{GrahamKnuthPatashnik} and~\cite[Thm~2.4]{Clark}). Note that $\binom{\sigma(m + 1)}{m + 1} \ge \binom{\sigma(m + 1)}{m}$, for $\sigma \ge 2$. Therefore, $\log(\frac{1}{\sigma(m + 1) + 1} \binom{\sigma(m + 1)}{m + 1}) \ge \log\binom{\sigma(m + 1)}{m} - O(\log (\sigma m))$ is a lower bound for the worst-case space consumption of any solution for the multiple pattern matching.

\newpage
\begin{theorem}
\label{theorem:main_theorem2}
For any constant $\varepsilon \in (0, 2)$ and any set $\dict$ of $d$ patterns over the alphabet $[0..\sigma)$ such that $d = o(m)$, where $m$ is the number of edges in the trie containing $\dict$, there is a data structure that allows to find all $occ$ occurrences of the patterns in text $T$ in $O(|T| + occ)$ time and occupies $L + \varepsilon m$ bits of space, where $L = \log\binom{\sigma(m + 1)}{m} - O(\log (\sigma m))$ is a lower bound on the worst-case space consumption for any such data structure.
\end{theorem}
\begin{proof}
Since $d = o(m)$, we have $d = \frac{m}{f(m)}$, where $f(m) \xrightarrow{m\to\infty} +\infty$, and therefore $d\log\frac{m}{d} = \frac{m}{f(m)}\log f(m) = o(m)$. Thus, applying Theorem~\ref{theorem:compact_tree} for $\frac{\varepsilon}{2}$, we obtain a data structure occupying $m H_k + 1.443 m + \frac{\varepsilon}{2}m + o(m)$ bits. Further, applying the simple encoding from Lemma~\ref{lemma:rrr} for the bit array $\numdictdef$ representing the transitions $\transdef$, we obtain a solution occupying $\log\binom{\sigma(m + 1)}{m} + \frac{\varepsilon}{2}m + o(m)$ bits. Since $\varepsilon$ is constant, we have $o(m) \le \frac{\varepsilon}{2}m$; hence, the result follows.
\end{proof}

\bibliography{paper}

\newpage
\appendix
\section{Implementation and Experiments}
\label{sec:implementation}

We implemented the data structure described in this paper in C++ and compared its runtime and memory consumption with the Belazzougui's solution~\cite{Belazzougui} and a naive algorithm. We could not find implementations of the Belazzougui's data structure and implemented it too. In~\cite{SokolShoshana} we found a different data structure based on compressed suffix trees but it showed a very poor performance, so we decided to exclude it from the tests.

The experiments were performed on a machine equipped with six 1.8 GHz Intel Xeon E5-2650L v3 CPUs with 30\,MiB L3 cache and 16\,GiB of RAM. The OS was Ubuntu 16.04.3 LTS, 64bit running kernel 4.4.0. All programs were compiled using {\tt g++} v5.4.0 with {\tt -O3} {\tt -march=x86-64} options. The source codes of all tested algorithms are available at \url{https://bitbucket.org/umqra/multiple-pattern-matching}. At the same URL one can find the 6 texts and 7 dictionaries on which the experiments were run.

The texts are as follows (see also Table~\ref{tab:data}):
\begin{itemize}
  \item {\tt program}: a sample binary file generated by a simple algorithm (see the URL above);
  \item {\tt chr1.dna}: Human first chromosome genome sequence in FASTA format\footnote{\url{ftp://ftp.ncbi.nih.gov/genomes/H_sapiens/Assembled_chromosomes/seq/hs_alt_CHM1_1.1_chr1.fa.gz}};
  \item {\tt wiki.en/wiki.ja/wiki.ru/wiki.zh}: the first 200Mb of the dump of all English/Japanese/Russian/Chinese wikipedia articles\footnote{\url{https://dumps.wikimedia.org/}}.
\end{itemize}
The dictionaries are as follows (see also Table~\ref{tab:data}):
\begin{itemize}
  \item {\tt dna.dict}: reads generated for {\tt chr1.dna} by {\tt wgsim} simulator\footnote{\url{https://github.com/lh3/wgsim}} with read length {\textasciitilde}100;
  \item {\tt urls.dict}: URLs from the .su/.nu zones\footnote{\url{https://zonedata.iis.se/}} and the Alexa database of popular URLs\footnote{\url{http://s3.amazonaws.com/alexa-static/top-1m.csv.zip}};
  \item {\tt virus.dict}: virus signatures from the {\tt main.cvd} file of the ClamAV database\footnote{\url{https://www.clamav.net/downloads}};
  \item {\tt ttl.en.dict/ttl.ja.dict/ttl.ru.dict/ttl.zh.dict}: the titles (in lower case) of some English/Japanese/Russian/Chinese wikipedia articles with length at least 3 letters\footnote{\label{foot:wiki}\url{https://dumps.wikimedia.org/}} (the dictionaries were truncated to reduce temporary space used in the index construction).
\end{itemize}

Our implementations use the {\tt SDSL} library by Gog~\cite{GogSDSL}. We tested the following algorithms:
\begin{itemize}
\item {\tt blz}: the original Belazzougui's compressed data structure~\cite{Belazzougui} in which the bitvectors $\numdictgen{0}, \ldots, \numdictgen{\sigma-1}$ were implemented using {\tt sd\_vector} from the {\tt SDSL};
\item {\tt cblz}: our algorithm with fixed block compression boosting;
\item {\tt cblz8}: the same as {\tt cblz} but with failure links sparsified using an 8-dense vertex subset;
\item {\tt smp}: a simple $O(nm)$-time algorithm that uses only two components of the {\tt blz} data structure: the bitvectors $\numdictgen{0}, \ldots, \numdictgen{\sigma-1}$ and an array of length $m{+}1$ encoding the $\flagdef$ flags; they all are implemented using {\tt sd\_vector} from the {\tt SDSL}.
\end{itemize}

For each algorithm, we ran 7 tests: in each test we search the patterns from a dictionary of Table~\ref{tab:data} in the corresponding text from the same table row (note that {\tt wiki.en} is used in two dictionaries). The results are present in Figure~\ref{fig:plot}. The Aho--Corasick data structure required too much memory in our experiments, so we do not include it.

\begin{table}[tb]
\caption{\label{tab:data}
Statistics of the test dictionaries of patterns (number of patterns, average pattern length, $m$, $\sigma$) and texts in which the patterns were searched (length, number of found occurrences).}\setlength\tabcolsep{5.2pt}
  \begin{tabular}[t]{lrrrr} \toprule
    Dictionary    & patterns  & avg. len. & $m$ & $\sigma$ \\\midrule
    dna.dict      & 1,999,911   &  100.0 & 178,323,409 & 5       \\
    virus.dict    & 4,059,198   &  16.0  & 56,430,521  & 256     \\
    urls.dict     & 3,825,132   &  16.7  & 39,385,319  & 40      \\
    ttl.en.dict   & 3,875,263   &  18.3  & 32,255,913  & 2881    \\
    ttl.ja.dict   & 1,691,693   &  8.4   & 7,182,594   & 7329    \\
    ttl.ru.dict   & 4,145,276   &  22.5  & 39,908,335  & 1350    \\
    ttl.zh.dict   & 1,649,209   &  7.3   & 6,164,034   & 12113   \\ \bottomrule
  \end{tabular}
  \begin{tabular}[t]{lrr} \toprule
    Text          & length         & found occ. \\\midrule
    chr1.dna      & 228,503,292    & 157,393     \\
    program       & 104,857,600    & 1,347,031   \\
    wiki.en       & 209,040,222    & 3,461,097   \\
    wiki.en       & 209,040,222    & 137,071,073 \\
    wiki.ja       & 109,018,155    & 15,196,791  \\
    wiki.ru       & 131,898,987    & 65,428,137 \\
    wiki.zh       & 126,021,996    & 18,925,337 \\ \bottomrule
  \end{tabular}
\end{table}

All dictionaries can be split into two unequal groups: dictionaries with long patterns (in our case it is only {\tt dna.dict}) and short patterns (all other dictionaries). On {\tt dna.dict} the fastest algorithms are {\tt blz} and {\tt cblz}; both {\tt cblz8} and {\tt smp} are about two times slower than them. Not surprisingly, {\tt smp} is the fastest algorithm on short patterns and the algorithms {\tt blz}, {\tt cblz}, and {\tt cblz8} are 2--3 times slower than {\tt smp}. Because of some implementation details (we use a specially tailored version of {\tt sd\_vector}), {\tt cblz} is faster than the simpler algorithm {\tt blz} in our tests. We were unable to explain why {\tt cblz8} often works faster than the simpler {\tt cblz} algorithm. To sum up, the three {\tt blz} algorithms have acceptable running times, but it really makes sense to use them only on relatively long dictionary patterns.

The central chart shows that in most cases the dictionaries are very well compressed (in comparison with $\frac{\log\sigma}{8}$ bytes per letter; see the shaded columns). But, as one can conclude from the low compression ratio of the corresponding trie on the right chart, the dictionary compression is mostly due to the assembling of the patterns in the trie (for instance, while the total length of the patterns $\{a^{i}b\}_{i=0}^{k-1}$ is $k(k+1) / 2$, their trie contains only $m = 2k$ edges). In the right chart, one can observe that the effect of compression boosting in {\tt cblz} is hidden under the overhead imposed by additional structures, and, because of this, {\tt cblz} occupies about the same space as {\tt blz}. The space optimized version {\tt cblz8} lowers the overhead and the resulting data structure occupies about the same space as {\tt smp}; this, surely, is possible only because {\tt cblz8} is $H_k$-compressed while {\tt smp} is only $H_0$-compressed.

\begin{figure}[hbt]
\centering
\hspace*{-0.7cm}\includegraphics[scale=0.38]{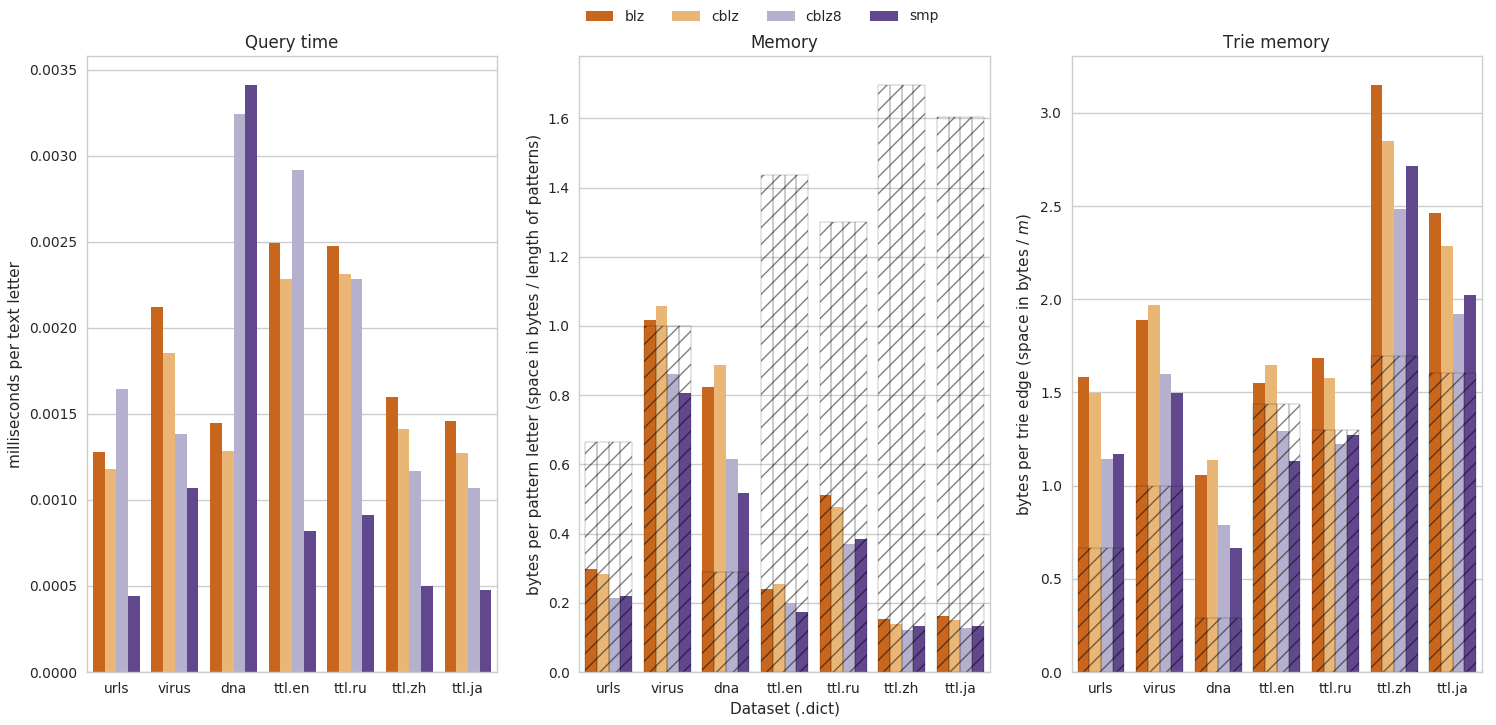}
\caption{Performance and space consumption; shaded columns correspond to $\frac{\log\sigma}{8}$.}\label{fig:plot}
\end{figure}

\end{document}